\documentclass[%
 reprint,
nofootinbib,
 amsmath,amssymb,
 aps,
]{revtex4-2}

\usepackage{algpseudocode}
\usepackage{algorithm}
\usepackage{url}
\usepackage{amsmath,amsthm}
\usepackage{graphicx}
\usepackage{booktabs}
\usepackage{physics}
\usepackage{amssymb}
\usepackage{tikz}
\usepackage{framed}

\usetikzlibrary{shapes.symbols}

\newtheorem{theorem}{\bf Theorem}
\newtheorem{proposition}{\bf Proposition}

\newtheorem{lemma}{\bf Lemma}
\newtheorem{definition}{\bf Definition}

\DeclareMathOperator*{\argmax}{arg\,max}

\newcommand{\X}{\mathbb{X}}
\newcommand{\Y}{\mathbb{Y}}

\allowdisplaybreaks

\begin{document}

\title{Maximal Information Leakage from Quantum Encoding of Classical Data}

\author{Farhad Farokhi}
 \email{farhad.farokhi@unimelb.edu.au}
\affiliation{Department of Electrical and Electronic Engineering, The University of Melbourne, Parkville, VIC 3010, Australia}

\date{\today}

\begin{abstract}
A new measure of information leakage for quantum encoding of classical data is defined. An adversary can access a single copy of the state of a quantum system that encodes some classical data and is interested in correctly guessing a general randomized or deterministic function of the data (e.g., a specific feature or attribute of the data in quantum machine learning) that is unknown to the security analyst. The resulting measure of information leakage, referred to as maximal quantum leakage, is  the multiplicative increase of the probability of correctly guessing any function of the classical data upon observing measurements of the quantum state. Maximal quantum leakage is shown to satisfy post-processing inequality (i.e., applying a quantum channel reduces information leakage) and independence property (i.e., leakage is zero if the quantum state is independent of the classical data), which are fundamental properties required for privacy and security analysis. It also bounds accessible information. Effects of global and local depolarizing noise models on the maximal quantum leakage are established. 
\end{abstract}


\maketitle

\section{Introduction}
This paper deals with quantifying the ``amount'' of classical information that can be leaked from a quantum system whose state encodes the said classical information. This is a basic question that arises in security and privacy analysis of quantum computing systems. In a setting where classical information is encoded into the state of a quantum system, e.g., quantum machine learning~\cite{biamonte2017quantum}, an adversary may either access the system legitimately or maliciously (by hacking), perform measurements on the state of the system, and extract private or secret information. This has, in part, motivated development of privacy-preserving quantum computing~\cite{zhou2017differential,hirche2023quantum,aaronson2019gentle,nuradha2023quantum,farokhi2023privacy}. This measure of information leakage can also be used to investigate security in communication or transmission of classical data over insecure quantum channels, where an eavesdropper may attempt to extract some classical information by performing measurements on the communicated qubits. This setup is akin to quantum wiretap channels~\cite{cai2004quantum}; see Section~\ref{sec:conc} for more information.

Any useful measure of information leakage must satisfy a few requirements~\cite{issa2019operational,farokhi2021measuring}. First, and foremost, the measure should possess an operational interpretation. As stated in~\cite[p.\,313]{wilde2013quantum}, ``the ultimate test for whether we truly understand an information measure is if it is the answer to some operational task.'' This will enable a designer or analyst to explain what guarantees can be extracted from minimizing or bounding the measure of information leakage. Second, assumptions regarding the adversary must be minimal so that a large family of adversaries can be modelled and analyzed. For instance, it is customary to assume that the adversary seeks to estimate the entire data accurately while, in practice, the adversary might only be seeking to extract as much information as possible\footnote{Akin to the so-called fishing expedition, which refers to non-specific search for information.} or might only be interested in estimating subsets of the data or specific features that may be unknown to the security analyst. Third, the measure should satisfy certain properties, such as post-processing inequality (i.e., further processing of the quantum system by an arbitrary quantum channel must reduce information leakage) and independence property (i.e., leakage is zero if the quantum state is independent of the classical data). The former enable the analyst to make statements that are independent of the computational power of the adversary, i.e., post processing by advanced computing techniques or powerful machines should not increase information leakage, while the latter implies that the measure of information leakage is not conservative, i.e., assigning a risk of information leakage to situations where the adversary is guaranteed to not gain any insight. Finally, the measure of information leakage should align with intuition, e.g., noisy quantum circuits must reduce the information leakage. 

Common measures of information leakage, while satisfying some of these properties, often fail to meet all requirements. For instance, in quantum computing and information theory, accessible information and its upper bound, Holevo information, do not meet the requirement on minimal assumptions on the adversary. They can be only used in the context that the adversary is interested in estimating the entirety of the classical data. This is because accessible information and Holevo information are formulated to study reliable information transmission~\cite{holevo1973bounds}, not information leakage in security analysis. Also, it is well-understood that mutual information is not a good measure of information leakage in security and privacy~\cite{issa2019operational}. Again, the operational interpretation of mutual information stems from communication and compression, which differs from security and privacy. In compression, for instance, we must be able to decode the entire information without any loss. However, in security, an adversary may not be interested in extracting the entire classical information. It may merely want to extract some or any private information. In the classical setting, these observations have motivated moving away from mutual information for measuring information leakage in privacy and security frameworks~\cite{liao2019tunable,diaz2019robustness,issa2019operational,farokhi2021measuring}.

This paper presents a new metric or measure for information leakage from quantum systems referred to as \textit{maximal quantum leakage}. The measure is built upon a similar classical notion of information leakage known as maximal leakage~\cite{issa2019operational}. The adversary has access to a single copy of the state $\rho_A^X$ of quantum system $A$ that encodes some classical data $X$, which is assumed to be private and must be kept secure. The adversary is interested guessing or estimating a general, possibly randomized function of the original classical data $X$, called $Z$. The adversary's intention or target, i.e., the underlying randomized function of the classical data, is not known to the designer or analyst. Therefore, the measure of information leakage must be maximized over all possible choices of this function. This motivates the use of the term \textit{maximal} in maximal quantum leakage. This threat model captures a large family of potential adversaries and thus minimizes the assumptions made regarding the adversary's intent. The adversary can perform measurements on the state of the quantum system to observe a random variable $Y$, i.e., the outcome of the measurements. The adversary then attempts to guess $Z$ based on $Y$ and verify whether the choice is correct. For instance, in the security framework, this could model guessing an individual's password and attempting to log in using the guess~\cite{issa2019operational}. However, other interpretations can be provided, e.g., this could capture guessing someone's private information information to use that for phishing attacks against them. The adversary's goal is to maximize the probability of the correctly guessing $Z$. That is, the adversary attempts to extract some information regarding $X$, modelled by $Z$, with high certainty. The information leakage measures the worst-case (i.e., maximal) ratio of the probability of correctly guessing $Z$ with access to $Y$ and without access to $Y$. Therefore, the measure investigate cases where the probability of correctly guessing $Z$ increases considerably based on access to $Y$, i.e., $Y$ leaks considerable information about $Z$, which is in turn a function of the private data $X$. This provides a natural interpretation for the information leakage: 
\begin{quote}
    The multiplicative increase in the probability of correctly guessing any general random or deterministic function of the private data upon access to the quantum encoding of the data is upper-bounded by the maximal quantum leakage.
\end{quote}
Important properties for the maximal quantum leakage is established. First, maximal quantum leakage is zero if the quantum encoding is indistinguishable, i.e., if the quantum state is independent of the classical data. This is a natural property as otherwise the measure of leakage acts conservative, i.e., it assigns a non-zero leakage to a scenario that possesses no risk. Second, maximal quantum leakage admits post-processing inequality, i.e., maximal leakage reduces if the state of the quantum system is manipulated by an arbitrary quantum channel. This is a useful property in privacy and security analysis because it implies that we only need to compute the information leakage at the beginning of the data analysis chain to establish the risk of data breach, and further computation cannot increase the risk. Finally, the effect of quantum noise models, such as global and local depolarizing channels, on the maximal quantum leakage are investigated. As expected, quantum noise reduces information leakage. This is not surprising given previous observations on the effect of noise in quantum devices on data privacy~\cite{zhou2017differential,hirche2023quantum}. However, establishing such results are important in ensuring that the proposed notion of information leakage accords with intuition. 

Before moving on to the technical content of the paper, a few remarks must be stated regarding the relationship between the framework of this paper, its classical counterpart~\cite{issa2019operational},  and the relevant literature on accessible information~\cite{davies1978information,Holevo2019,vrehavcek2005iterative}. The main difference with the classical results in~\cite{issa2019operational} stems from the fact that the conditional probability of measurements given private classical data can be written explicitly in terms of quantum states (i.e., the quantum encoding of the classical data) and positive operator-valued measure (POVM) modelling the measurement process using the Born's rule. This implies a degree of freedom that is missing in the classical counterpart. That is, in~\cite{issa2019operational}, the conditional probabilities are fixed, but, in this paper, the adversary can changes the conditional probabilities by varying the POVM, i.e., the adversary can select the optimal measurement process for extracting as much information as possible. This implies an additional optimization over the POVMs, which can be potentially unbounded. To alleviate this difficulty, we use methods developed for accessible information~\cite{davies1978information,Holevo2019, vrehavcek2005iterative} to show that the number elements in the POVM is bounded and can be computed iteratively.

The remainder of the paper is organized as follows. We finish this section with some preliminary definitions and useful notations for quantum systems. Maximal quantum leakage is formally defined and a semi-explicit formula for its computation is presented in Section~\ref{sec:leakage}. Properties of the maximal quantum leakage, i.e., independence property, post-processing inequality, and upper and lower bounds for the leakage, are established in Section~\ref{sec:properties}. The effect of depolarizing noises inherent to quantum devices on maximal quantum leakage are investigated in Section~\ref{sec:noise}. To minimize interruptions to the flow of the paper and to focus on definitions and properties without getting bogged down in the mathematics, the proofs of all the results in Sections~\ref{sec:leakage}--\ref{sec:noise} are presented across different subsections in Appendix~\ref{sec:proofs}. Finally, Section~\ref{sec:conc} presents some concluding remarks and avenues for future research.

\subsection{Preliminary Definitions and Notations} 
The basic definitions and useful notations presented in this review section are mostly borrowed from~\cite{wilde2013quantum}. 

The state space of a quantum system is modelled by complex Hilbert space $\mathcal{H}$. Dirac's notation is used to denote pure quantum states. A \textit{pure quantum state} is an element of Hilbert space $\mathcal{H}$ with unit norm, e.g., $\ket{\psi}\in\mathcal{H}$ with $\|\ket{\psi}\|_2=\sqrt{\braket{\psi}{\psi}}=1$. The smallest non-trivial quantum system is a \textit{qubit} corresponding to a 2-dimensional Hilbert space. Combination of any two quantum states $\ket{\phi}$ and $\ket{\psi}$ is denoted by their tensor product $\ket{\phi}\otimes\ket{\psi}$. A \textit{mixed quantum state} is characterized by ensemble $\{(p_1,\ket{\psi_1}),\dots,(p_k,\ket{\psi_k})\}$, where $p_i\geq 0$ for all $i\in\{1,\dots,k\}$ and $\sum_{i} p_i=1$. The mixed state signifies that the quantum system is in pure state $\ket{\psi_i}$ with probability $p_i$. The density operator for the mixed quantum state $\{(p_1,\ket{\psi_1}),\dots,(p_k,\ket{\psi_k})\}$ is $\rho:=\sum_{i} p_i\ket{\psi_i}\bra{\psi_k}$. Thus, $\trace(\rho)=1$. Any pure quantum state $\ket{\phi}$ can be modelled using rank-one density operator $\rho=\ket{\phi}\bra{\phi}$. Therefore, without loss of generality, the density operator can denote the state of a quantum system. Combination of any two density operators $\rho$ and $\sigma$ is denoted by their tensor product $\rho\otimes\sigma$.

When the post-measurement state of the quantum system is of no interest (e.g., the quantum system is discarded after measurement), a measurement for a quantum system can be modelled using a positive operator-valued measure (POVM), which is a set of positive semi-definite matrices $F=\{F_i\}_{i}$ such that $\sum_{i} F_i=I$. In this case, the probability of obtaining output $i$ when taking a measurement on a system with quantum state $\rho$ is $\trace(\rho F_i)=\trace( F_i \rho)$. 


A quantum channel is a mapping from the space of density operators to potentially another space of density operators that is both completely positive and trace preserving. Following the Choi-Kraus theorem~\cite[Theorem~4.4.1]{wilde2013quantum}, for each quantum channel $\mathcal{E}$, there exists a family of linear operators $\{E_j\}_{j}$ satisfying $\sum_{j}E_j^\dag E_j=I$ such that $\mathcal{E}(\rho)=\sum_{j} E_j\rho E_j^\dag$ for all density operators $\rho$. For any matrix $A$, $A^\dag$ denotes its conjugate transpose or Hermitian. Tensor product of quantum channels $\mathcal{E}_1$ and $\mathcal{E}_2$ is defined as $\mathcal{E}_1\otimes\mathcal{E}_2(\rho_1\otimes\rho_2):=\mathcal{E}_1(\rho_1)\otimes\mathcal{E}_2(\rho_2)$ for all density operators $\rho_1$ and $\rho_2$. 

\section{Definition and Computation of Maximal Quantum Leakage} \label{sec:leakage}
The classical data that must be kept private or secure is modelled by discrete random variable $X$ with finite support set $\X$. Assume that $p_X(x):=\mathbb{P}\{X=x\}>0$ for all $x\in\mathbb{X}$. This assumption is without loss of generality as the set $\X$ can always be trimmed so that this assumption holds. Knowledge of the support set of secrete variable is referred to as \textit{domain knowledge} and is postulated to be required for developing privacy-preserving mechanisms~\cite{nuradha2023quantum}. As shown later, the maximal quantum leakage is not a function of $p_X$ is thus robust to the choice of the secret prior. This aligns with the requirement to keep the assumptions on the measure of information leakage minimal. 

For each realization of discrete random variable $X=x\in\X$, a quantum system $A$ in mixed state $\rho_A^x\in\mathcal{D}(\mathcal{H}_A)$ is prepared, i.e., ensemble $\mathcal{E}:=\{p_X(x),\rho_A^x\}_{x\in\X}$ is prepared. The quantum state $A$ is handed over to an adversary without revealing the realization of the classical random variable $X$. The expected density operator is then $\rho_A=\mathbb{E}\{\rho_A^X\}=\sum_{x\in\X}p_X(x)\rho_A^x$. This is the state of the quantum system from the perspective of someone who does not know the realization of $X$, i.e., the adversary. 

The objective of the adversary is estimate or guess a possibly randomized discrete function of the random variable $X$, denoted by random variable $Z$, by performing measurements on the quantum system $A$. The adversary performs positive operator-valued measure (POVM) $F=\{F_y\}_y$ on the quantum system $A$. Random variable $Y$ with finite support set $\Y=\{1,\dots,|F|\}$ denotes the outcome of the measurement. The probability of obtaining output $Y=y\in\Y$ when taking a measurement on quantum state $\rho_A^X$ is given by $\trace(\rho_A^X F_i)$. Therefore, 
\begin{align} \label{eqn:conditional}
    \mathbb{P}\{Y=y\,|\,X=x\}=\trace(\rho_A^x F_y),\quad \forall x\in\X,y\in\Y.
\end{align}
Upon observing the measurement outcome $Y$, the adversary takes a one-shot guess of the random variable $Z$ denoted by the random variable $\widehat{Z}$. The adversary then attempts to verify whether $\widehat{Z}$ is correct or not. 

\begin{definition}[Maximal Quantum Leakage] \label{def:qml} The maximal quantum leakage from random variable $X$ through quantum encoding of the data via ensemble $\{p_X(x),\rho_A^x\}_{x\in\X}$ is 
\begin{align}
    \mathcal{Q}(X\rightarrow A)_{\rho_A}:=\sup_{\{F_y\}_y}\sup_{Z,\widehat{Z}} \log_2 \left(\frac{\mathbb{P}\{Z=\widehat{Z}\}}{\displaystyle \max_{z\in\mathbb{Z}}\mathbb{P}\{Z=z\}} \right),
\end{align}
    where the inner supremum is taken over all random variables $Z$ and $\widehat{Z}$ with arbitrary support set $\mathbb{Z}$ and the outer supremum is taken over all POVMs $F=\{F_y\}_y$. 
\end{definition}

The maximal quantum leakage, as characterized in Definition~\ref{def:qml}, captures the multiplicative increase in the probability of correctly guessing any general random or deterministic function of the private data upon accessing the quantum encoding of the data. The probability of correctly guessing the realization of random variable $Z$ with access to measurement $Y$ is $\mathbb{P}\{Z=\widehat{Z}\}$. Without access to any measurements, the adversary's best guess of the realization of random variable $Z$ would have been the most likely or probable realization $\widetilde{Z}:=\argmax_{z\in\mathbb{Z}} \mathbb{P}\{Z=z\}$. Therefore, the probability of correctly guessing  the realization of random variable $Z$ without access to any measurements is $\mathbb{P}\{Z=\widetilde{Z}\}=\max_{z\in\mathbb{Z}}\mathbb{P}\{Z=z\}$. A large maximal quantum leakage implies that there exist features of the private data that can be guessed more reliably by accessing the quantum state. This demonstrates information leakage along those features. Noting that those features can be potentially exploited by the adversary, an analyst, who is  not aware of the target of the adversary, has to investigate and mitigate the weak point. 

\begin{theorem} \label{tho:mql} The maximal quantum leakage is given by
    \begin{align} \label{eqn:maxim_leakage_formula}
    \mathcal{Q}(X\rightarrow A)_{\rho_A}=\sup_{\{F_y\}_y} &\log_2\left(\sum_{y\in\mathbb{Y}} \max_{x\in\mathbb{X}} \trace(\rho_A^x F_y) \right).
 \end{align}
\end{theorem}

\begin{proof}
See Appendix~\ref{proof:tho:mql}.
\end{proof}

Theorem~\ref{tho:mql} provides a semi-explicit formula for computing maximal quantum leakage. This is done by removing the first maximization over random variables $Z$ and $\widehat{Z}$. Interestingly, Theorem~\ref{tho:mql} shows that $\mathcal{Q}(X\rightarrow A)_{\rho_A}$ is independent of the prior for the secret $p_X$. Therefore, maximal quantum leakage is immune from or robust to incorporating a wrong assumption on the secret random variable $X$. The same cannot be said about accessible information, quantum mutual information, or Holevo information. In Theorem~\ref{tho:mql}, however, the outer maximization on $\{F_y\}$ still remains. This is a particularly troubling problem as the number of outcomes in POVM $\{F_y\}$ is not bounded, i.e., it can range to infinity. The next theorem shows that we can restrict our search to POVMs that have at most $\dim(\mathcal{H}_A)^2$ outcomes. Note that, because we cannot exchange the sum over $y$ and the maximization over $x$ in~\eqref{eqn:maxim_leakage_formula}, the information
leakage is nonzero in general.

\begin{theorem} \label{tho:maxsize}
    Let $\mathcal{H}_A$ have finite dimension $d_A$. The maximal quantum leakage is attained by POVM $F=\{F_y\}_{y=1}^m$ with $m\leq d_A^2$, such that $F_y$ are rank-one operators. 
\end{theorem}

\begin{proof}
    See Appendix~\ref{proof:tho:maxsize}.
\end{proof}

Maximal quantum leakage in Definition~\ref{def:qml} considers a scenario where the adversary only makes and verifies a single guess. This might not be entirely realistic in practice. An adversary might be able to make several guesses. For instance, the adversary might devise multiple privacy or security attacks based on various plausible guesses of the random variable $Z$. Therefore, we may assume that, upon observing the measurement outcome $Y$, the adversary makes $k$ guess of the random variable $Z$ denoted by the random variable $\widehat{Z}_1,\dots,\widehat{Z}_k$ and then attempts to verify them. In this scenario, we can modify maximal quantum leakage to compute $k$-maximal quantum leakage, defined below.

\begin{definition}[$k$-Maximal Quantum Leakage] \label{def:qml_k} The $k$- guess maximal quantum leakage from random variable $X$ through quantum encoding of the data via ensemble $\{p_X(x),\rho_A^x\}_{x\in\X}$ is 
\begin{align}
    \mathcal{Q}^{(k)}&(X\rightarrow A)_{\rho_A}\nonumber\\
    &:=\sup_{\{F_y\}_y}\sup_{Z,\widehat{Z}_1,\dots,\widehat{Z}_k} \log_2 \left(\frac{\mathbb{P}\{\exists j:Z=\widehat{Z}_j\}}{\displaystyle \max_{\mathcal{Z}\subset\mathbb{Z},|\mathcal{Z}|\leq k}\mathbb{P}\{Z\in\mathcal{Z}\}} \right),
\end{align}
    where the inner supremum is taken over all random variables $Z,\widehat{Z}_1,\dots,\widehat{Z}_k$ with arbitrary support set $\mathbb{Z}$ and the outer supremum is taken over all POVMs $F=\{F_y\}_y$. 
\end{definition}

\begin{theorem} \label{tho:k-leakage}
    $\mathcal{Q}^{(k)}(X\rightarrow A)_{\rho_A}=\mathcal{Q}(X\rightarrow A)_{\rho_A}$. 
\end{theorem}

\begin{proof}
    The proof follows from~\cite[Theorem~4]{issa2019operational} and the proof of Theorem~\ref{tho:mql}.
\end{proof}

Theorem~\ref{tho:k-leakage} implies that number of guesses that the adversary can make is immaterial in measuring information leakage. Therefore, the choice of one-shot guesses in maximal quantum leakage is without loss of generality. 

\section{Properties of Maximal Quantum Leakage} \label{sec:properties}
In this section, properties of maximal quantum leakage are established. Maximal quantum leakage satisfies independence property (i.e., leakage is zero if the quantum state is independent of the classical data) and post-processing inequality (i.e., applying a quantum channel reduces information leakage). We can also bound maximal quantum leakage based on the dimension of the quantum system and the cardinality of the support set of the secret random variable. We start with independence property. To do so, we must define indistinguishably to establish when the quantum state is independent of the classical data. 

\begin{definition}[Indistinguishability] $(\rho_A^x)_{x\in\mathbb{X}}$ is indistinguishable if $\rho_A^x=\rho_A^{x'}$ for all $x,x'\in\X$. 
\end{definition}

Indistinguishability implies that, for various realizations of the classical data $X$, the quantum state remains the same. Therefore, an adversary cannot obtain any measurements from the quantum states that correlate with the classical data. Therefore, there is no leakage of classical data. This is established in the next result.

\begin{proposition} \label{prop:independence}
    $\mathcal{Q}(X\rightarrow A)_{\rho_A}\geq 0$ with equality if and only if $(\rho_A^x)_{x\in\mathbb{X}}$ is indistinguishable. 
\end{proposition}

\begin{proof}
    See Appendix~\ref{proof:prop:independence}.
\end{proof}

In the next proposition, we provide an upper bound for maximal quantum leakage based on the dimension of the quantum system $d_A=\dim(\mathcal{H}_A)$ and the cardinality of the support set of the secret random variable $|\X|$. A discrete random variable $\X$ with support set of size $|\X|$  has no more than $\log_2(|\X|)$ bits of information to be leaked. Therefore, $\mathcal{Q}(X\rightarrow A)_{\rho_A}$ can never be larger than $\log_2(|\X|)$. 

\begin{proposition} \label{property:2}
    Let $\mathcal{H}_A$ have finite dimension $d_A$. Then, $\mathcal{Q}(X\rightarrow A)_{\rho_A}\leq \min\{\log_2(|\X|),\log_2(d_A^2)\}$.
\end{proposition}

\begin{proof}
    See Appendix~\ref{proof:property:2}.
\end{proof}

Another important property is the data-processing inequality stating that quantum maximal leakage can be only reduced by application of an arbitrary quantum channel. This implies that we only need to compute the information leakage at the beginning of the data analysis chain to establish the risk of data breach, and further computation cannot increase the risk. The data-processing inequality is proved in the next proposition.

\begin{proposition} \label{prop:dataprocessing}
    For any quantum channel $\mathcal{E}$, $\mathcal{Q}(X\rightarrow A)_{\mathcal{E}(\rho_A)}\leq \mathcal{Q}(X\rightarrow A)_{\rho_A}$. 
\end{proposition}

\begin{proof}
    See Appendix~\ref{proof:prop:dataprocessing}.
\end{proof}

An important notion of information in quantum information theory is accessible information~\cite[p.\,298]{wilde2013quantum}. For ensemble $\mathcal{E}=\{p_X(x),\rho_A^x\}_{x\in\X}$, defined in Section~\ref{sec:leakage}, the accessible information is 
\begin{align*}
    I_{acc}(\mathcal{E}):=\sup_{\{F_y\}_y} I(X;Y),
\end{align*}
where $I(X;Y)$ is the classical mutual information between the random variable $X$ denoting the secret or private information and the random variable $Y$ denoting the measurement outcome. The next proposition provides a relationship between accessible information and maximal quantum leakage. 

\begin{proposition} \label{prop:relation_accessible_info}
    $I_{acc}(\mathcal{E})\leq \mathcal{Q}(X\rightarrow A)_{\rho_A}$. 
\end{proposition}

\begin{proof}
    See Appendix~\ref{proof:prop:relation_accessible_info}.
\end{proof}

The inequality in Proposition~\ref{prop:relation_accessible_info} is rather intuitive. The accessible information, in the context of security analysis, deals with an adversary that seeks to estimate the entire secret data. However, in defining maximal quantum leakage, we let the adversary extract as much information as possible by estimating any general possibly randomized function of the data. The adversary in the maximal quantum leakage setting is stronger and more general in comparison with the adversary in the context of the accessible information. This inequality is a direct consequence of requiring minimal assumptions on the adversary.

\section{Effect of Quantum Noise on Maximal Leakage} \label{sec:noise}

A common noise model in quantum systems is the (global) depolarizing channel defined as
\begin{align} \label{eqn:dep_channel}
    \mathcal{D}_{p,d_A}(\rho):=\frac{p}{d_A}I+(1-p)\rho,
\end{align}
where $d_A$ is the dimension of the Hilbert space $\mathcal{H}_A$ to which the system belongs and $p\in[0,1]$ is a probability parameter. In the next proposition, it is shown that the depolarizing channel results in reduction of the maximal quantum leakage. This means that maximal quantum leakage accords with intuition and similar results on privacy-preserving quantum systems~\cite{zhou2017differential,hirche2023quantum,aaronson2019gentle,nuradha2023quantum,farokhi2023privacy}. 

\begin{proposition} \label{prop:global}
For global depolarizing channel $\mathcal{D}_{p,d_A}$,
    $$\mathcal{Q}(X\rightarrow A)_{\mathcal{D}_{p,d_A}(\rho_A)}=\log_2( p +(1-p)2^{\mathcal{Q}(X\rightarrow A)_{\rho_A}}).$$
Particularly,
    \begin{align*}
        \frac{\mathrm{d}}{\mathrm{d}p}\mathcal{Q}(X\rightarrow A)_{\mathcal{D}_{p,d_A}(\rho_A)}<0 \mbox{ if } \mathcal{Q}(X\rightarrow A)_\rho>0.
    \end{align*}
Therefore, $\mathcal{Q}(X\rightarrow A)_{\mathcal{D}_{p,d_A}(\rho_A)}$ is a decreasing function of the probability parameter $p$.
\end{proposition}

\begin{proof} 
See Appendix~\ref{proof:prop:global}.
\end{proof}

\begin{figure}
    \centering
    \begin{tikzpicture}
    \node[] at (0,0) {\includegraphics[trim=1.3cm .1cm 1.2cm .5cm,clip,width=.9\columnwidth]{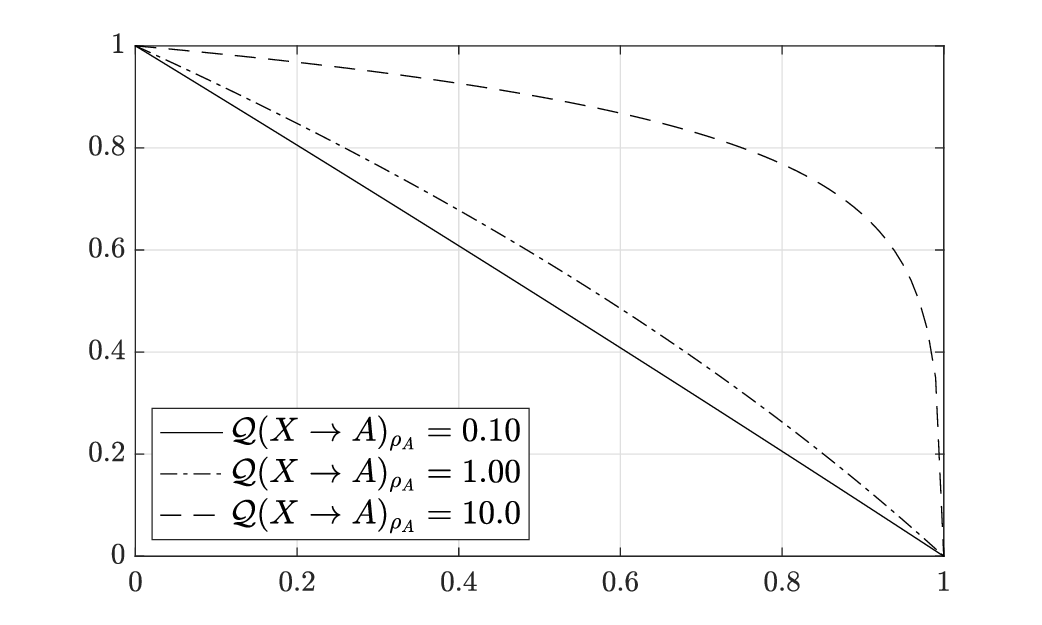}};
    \node[] at (0,-2.8) {$p$};
    \node[rotate=90] at (-4.2,.2) {$\mathcal{Q}(X\rightarrow A)_{\mathcal{D}_{p,d_A}(\rho_A)}/\mathcal{Q}(X\rightarrow A)_{\rho_A}$};
\end{tikzpicture}
\vspace{-3mm}
    \caption{Ratio of information leakage without and with global depolarizing channel versus the probability parameter $p$. As expected, the noisier the channel is, i.e., the higher the probability parameter is, the smaller maximal quantum leakage is. }
    \label{fig:global}
\end{figure}

Fig.~\ref{fig:global} illustrates the ratio of information leakage without and with global depolarizing channel versus the probability parameter. When the probability parameter rises, and therefore the global depolarizing channel becomes noisier, the maximal quantum leakage drops continuously. 

The previous proposition demonstrated how maximal quantum leakage is affected by global depolarizing noise. However, in quantum computing devices, each qubit can be affected by local noise. Consider the case where the Hilbert space $\mathcal{H}_A$ is composed of $k$ qubits, i.e., $d_A=2^k$. In this case, we consider local depolarizing noise channel $\mathcal{D}_{p,2}^{\otimes k}:=\mathcal{D}_{p,2}\otimes\cdots\otimes \mathcal{D}_{p,2}$, where a depolarising channel $\mathcal{D}_{p,2}$ acts on each qubit separately. The effect of local depolarizing channel on the maximal quantum leakage is investigated in the next proposition.

\begin{proposition} \label{prop:local}
For local depolarizing channel $\mathcal{D}_{p,2}^{\otimes k}$,
    $$\mathcal{Q}(X\rightarrow A)_{\mathcal{D}_{p,2}^{\otimes k}(\rho_A)}\leq \log_2( p^k +(1-p^k)2^{\mathcal{Q}(X\rightarrow A)_{\rho_A}}).$$
\end{proposition}

\begin{proof} 
See Appendix~\ref{proof:prop:local}.
\end{proof}

Propositions~\ref{prop:global} and~\ref{prop:local} show that noisy intermediate-scale quantum (NISQ) devices inherently provide security and privacy. This is of course not surprising. The noise in  NISQ devices has been shown to ensure quantum differential privacy~\cite{zhou2017differential, hirche2023quantum} and privacy against hypothesis testing adversaries~\cite{farokhi2023privacy}. Noisy devices can also improve security of quantum machine learning models against adversarial attacks~\cite{du2021quantum,weber2021optimal}. Albeit these guarantees go hand in hand with performance degradation~\cite{resch2021benchmarking}. 

\begin{algorithm}[H]
    \caption{\label{alg:1} Subgradient ascent algorithm for computing maximal quantum leakage. }
    \begin{algorithmic}[1]
        \Require $\{\rho_A^x\}_{x\in\X}$, $\Y=\{1,\dots,d_A^2\}$, $\mu>0$, and $\epsilon>0$
        \Ensure $\{F_y\}$ and $\mathcal{Q}(X\rightarrow A)_{\rho_A}$
        \State $\mathrm{OldCost} \gets \infty$
        \State Random initialization $F_y\succeq 0$, $\forall y\in\Y$
        \State $\mathrm{NewCost} \gets \sum_{y\in\mathbb{Y}} \max_{x\in\X}\trace(\rho_A^x F_y)$
        \While {$|\mathrm{OldCost}-\mathrm{NewCost}|\geq\epsilon$}
            \State $\mathrm{OldCost} \gets\mathrm{NewCost} $
            \State $x^*(y)\gets \argmax_{x\in\X} \trace(\rho_A^x F_y)$, $\forall y\in\Y$
            \State $G_y\gets I+\mu \left(\rho_A^{x^*(y)}-\sum_{z\in\Y} \rho_A^{x^*(z)}F_z\right)$, $\forall y\in\Y$
            \State $\tilde{F}_y \gets G_y^\dag F_y G_y$
            \State $S\gets \sum_{y\in\Y} \tilde{F}_y$
            \State $F_y\gets S^{-1/2}\tilde{F}_yS^{-1/2}$
            \State $\mathrm{NewCost} \gets \sum_{y\in\mathbb{Y}} \max_{x\in\X}\trace(\rho_A^x F_y)$
        \EndWhile
        \State $\mathcal{Q}(X\rightarrow A)_{\rho_A}=\log_2(\mathrm{NewCost})$
        \State \Return $\{F_y\}$ and $\mathcal{Q}(X\rightarrow A)_{\rho_A}$
    \end{algorithmic}
\end{algorithm}

\section{Iterative Algorithm for Computing Maximal Quantum Leakage}
In this section, we follow the approach of~\cite{vrehavcek2005iterative} to compute the maximal quantum leakage using an iterative algorithm. The main difference here is the use of subgradient (as opposed to gradient) ascent. This is due to non-differentiability of the cost function in maximal quantum leakage with respect to the POVMs (due the inner maximization on $x\in\X$); see~\cite[\S\,14.2-14.3]{sun2006optimization} for more information on subgradients and non-smooth optimization. Note that
\begin{subequations}\label{eqn:update_problem}
    \begin{align}
    2^{\mathcal{Q}(X\rightarrow A)_{\rho_A}}=\sup_{\{F_y\}} &\sum_{y\in\mathbb{Y}}  \trace(\rho_A^{x^*(y)} F_y),\\
   \mathrm{s.t.}\; & 0\preceq F_y,y\in\Y, \sum_{y\in\Y} F_y=I,
\end{align}
\end{subequations}
where $\Y=\{1,\dots,d_A^2\}$, $x^*(y)\in \argmax_{x\in\X} \trace(\rho_A^x F_y)$, and  $d_A$ is the dimension of $\mathcal{H}_A$. If $\argmax_{x\in\X} \trace(\rho_A^x F_y)$ admits a unique solution for all $y\in\Y$, the cost function is differentiable at that point and the subgradient and the gradient are equal to each other. If $\argmax_{x\in\X} \trace(\rho_A^x F_y)$ does not admit a unique solution, selecting each solution results in a different subgradient. In fact, any convex combination of said subgradients will  be also a subgradient. Note that after fixing $x^*(y)$, the cost function and constraints of this optimization problem in~\eqref{eqn:update_problem} has the same form as in~\cite{vrehavcek2005iterative} and, therefore, a similar approach can be used to compute the subgradients (instead of gradients). Algorithm~\ref{alg:1} summarizes an iterative subgradient ascent method for computing the maximal leakage. In this algorithm, $\mu>0$ is the step size and can be selected adaptively to keep the cost function increasing or can be set \textit{a priori} small enough to ensure convergence. Furthermore, threshold $\epsilon>0$ is selected to determine termination of the algorithm, i.e., the algorithm is terminated if the improvements in the cost function is not larger than $\epsilon$ (so not enough headway towards an optimal solution is made).

\begin{figure}
\centering
\begin{tikzpicture}
    \node[] at (0,0) {\includegraphics[trim=1.3cm .1cm 1.2cm .5cm,clip,width=.9\columnwidth]{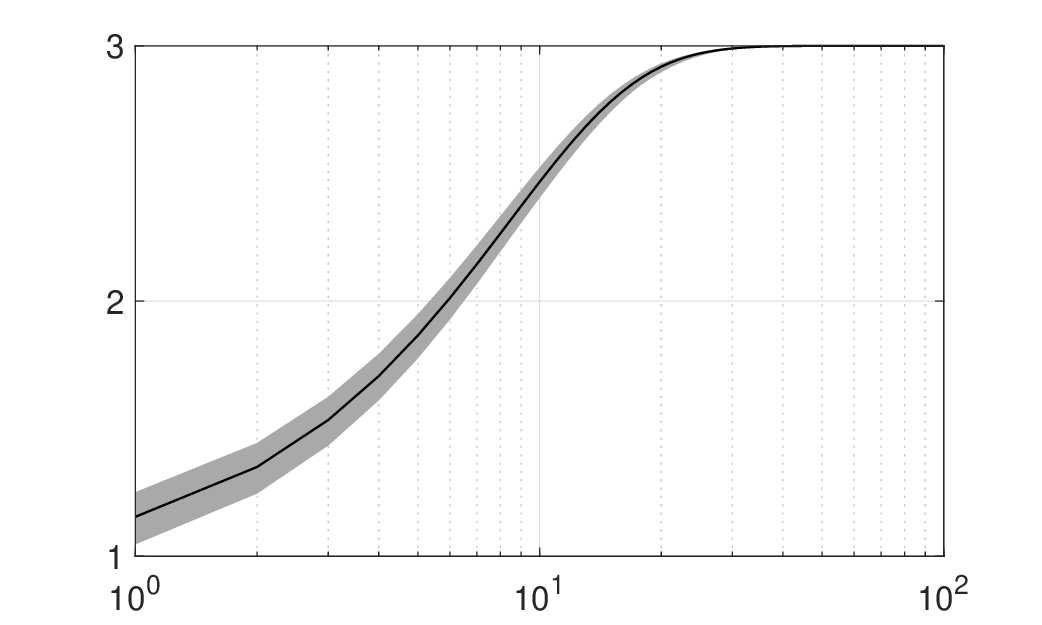}};
    \node[] at (0,-2.8) {Iteration number};
    \node[rotate=90] at (-4.2,0) {$\log_2(\sum_{y\in\mathbb{Y}} \max_{x\in\mathbb{X}} \trace(\rho_A^x F_y))$};
\end{tikzpicture}
\vspace{-3mm}
\caption{Information leakage versus iterations of Algorithm~\ref{alg:1} for index encoding example.}
\label{fig:index}
\end{figure}

To demonstrate the validity of the proposed algorithm, we start by index encoding of classical data in quantum states for which maximal quantum leakage can be computed easily. Consider random variable $X$ with support set $\X=\{1,\dots,d_A\}$, where $d_A=8$ denotes the dimension of $\mathcal{H}_A$. Let $\rho_A^x=\ket{x}\bra{x}$ for all $x\in\X$. In this case, by selecting $F_y=\ket{y}\bra{y}$ for $y\in\{1,\dots,d_A\}$, we get $\log_2(\sum_{y\in\mathbb{Y}}  \trace(\rho_A^{x^*(y)} F_y))=\log_2(d_A)=3$. This is the maximum attainable leakage according to Proposition~\ref{property:2}. Therefore, $\mathcal{Q}(X\rightarrow A)_{\rho_A}=3$ bits. Figure~\ref{fig:index} illustrates the information leakage for iterations of Algorithm~\ref{alg:1} starting from a random POVM. Following Theorem~\ref{tho:maxsize}, we select $d_A^2=64$ as the number of elements in the starting POVM. Furthermore, we use $\mu=10^{-1}$. The grey area demonstrates the maximum and minimum in each iteration (note the randomness in the initialization) and solid black line shows the mean in each iteration. As expected, the algorithm rapidly converges to $\mathcal{Q}(X\rightarrow A)_{\rho_A}=3$. 

\begin{figure}
\centering
\begin{tikzpicture}
    \node[] at (0,0) {\includegraphics[trim=1.3cm .1cm 1.2cm .5cm,clip,width=.9\columnwidth]{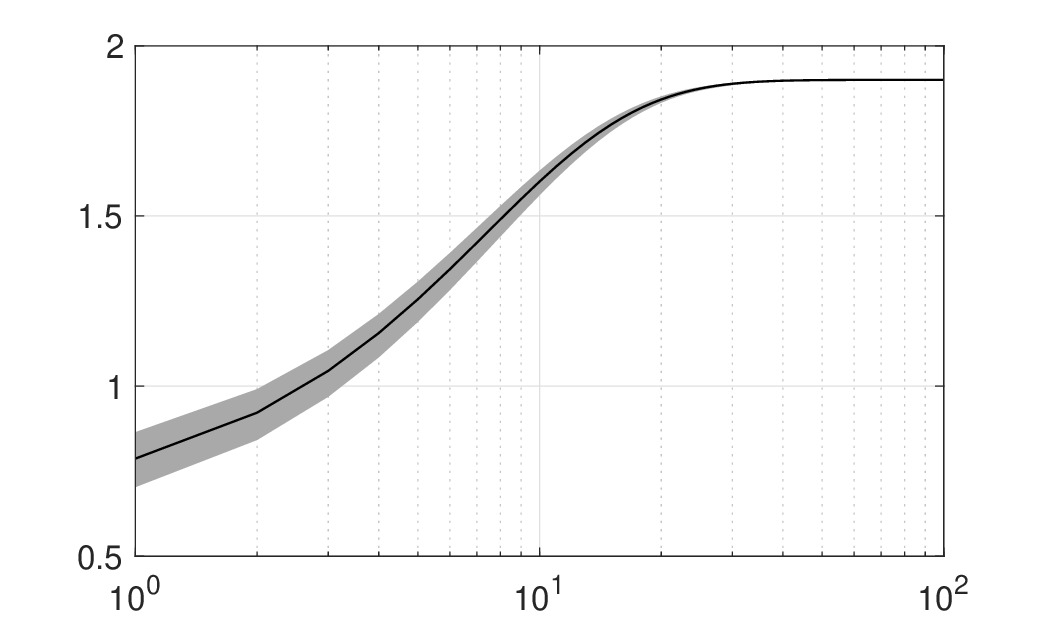}};
    \node[] at (0,-2.8) {Iteration number};
    \node[rotate=90] at (-4.2,0) {$\log_2(\sum_{y\in\mathbb{Y}} \max_{x\in\mathbb{X}} \trace(\rho_A^x F_y))$};
\end{tikzpicture}
\vspace{-3mm}
\caption{Information leakage versus iterations of Algorithm~\ref{alg:1} for amplitude encoding example.}
\label{fig:amplitude}
\end{figure}

Now, we expand our attention to a more complex encoding strategy. Let $X=(X_1,X_2,X_3)\in\{0,1\}^3$. Assume that $\mathcal{H}_A$ is a Hilbert space of dimension $d_A=8$ and $\rho_A^x=\ket{\psi^x}\bra{\psi^x}$, where $\psi^x=x_1\ket{0}+(1-x_1)\ket{1}+x_2\ket{2}+(1-x_2)\ket{3}+x_3\ket{4}+(1-x_3)\ket{5}$. This is more similar to amplitude encoding that is often utilized in quantum machine learning. Figure~\ref{fig:amplitude} shows the information leakage versus iterations of Algorithm~\ref{alg:1} starting from a random POVM. Similarly, we use $\mu=10^{-1}$ and select $d_A^2=64$ as the number of elements in the starting POVM. The grey area demonstrates the maximum and minimum in each iteration  and solid black line shows the mean in each iteration. The algorithm rapidly converges to $\mathcal{Q}(X\rightarrow A)_{\rho_A}=1.9$ bits. Interestingly, amplitude encoding seems to leak less information in comparison with index encoding.

\section{Conclusions and Future Work} \label{sec:conc}
We considered an adversary that is interested in correctly guessing a potentially randomized function of a secret or private data with access to a single copy of the state of a quantum system encoding it. We proposed the novel notion of maximal quantum leakage, which captures the multiplicative increase in the probability of correctly guessing any function of the data upon observing measurements of the quantum state. We proved that maximal quantum leakage satisfies post-processing inequality, independence property, and bounds accessible information. Future work can focus on the following topics:

\begin{figure}
    \centering
    \begin{tikzpicture}
        \node[rectangle,draw,minimum height=1cm] (1) at (0.3,0) {Encoder};
        \node[cloud, draw,minimum height=1cm,minimum width=2.2cm,cloud puffs = 16] (2) at (3,0) {};
        \node[] at (3,0) {Channel 1};
        \node[cloud, draw,minimum height=1cm,minimum width=2.2cm,cloud puffs = 16] (3) at (3,-1.5) {};
        \node[] at (3,-1.5) {Channel 2};
        \draw[->] (-1.2,0) -- (1);
        \draw[->] (1) -- (2);
        \draw[->] (1.5,0) -- (1.5,-1.5) -- (3);
        \draw[->] (2) -- (5,0);
        \draw[->] (3) -- (5,-1.5);
        \node[] at (-1.7,0.05) {Alice};
        \node[] at (-.85,0.2) {$X$};
        \node[] at (1.5,0.2) {$\rho_A^x$};
        \node[] at (4.5,0.2) {$\rho_B^x$};
        \node[] at (5.4,0.05) {Bob};
        \node[] at (4.5,-1.3) {$\rho_E^x$};
        \node[] at (5.4,-1.45) {Eve};
    \end{tikzpicture}
    \caption{A quantum wiretap channel, where Alice wants to communicate effectively with Bob while minimizing the leaked information to Eve.}
    \label{fig:cipher}
\end{figure}
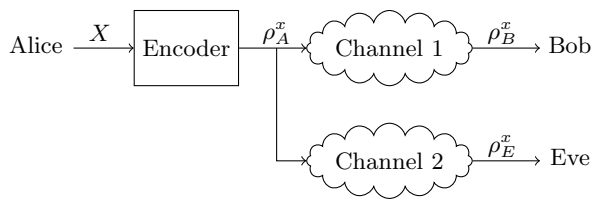

\begin{itemize}
    \item \textit{Quantum Wiretap}: 
    As stated in the introduction, accessible information (related to mutual information) is an appropriate measure of information when considering lossless communication~\cite{holevo1973bounds} while maximal quantum leakage generalizes the assumptions on the adversary and is perfect for investigating eavesdroppers. This motivates using a combination of accessible information and maximal quantum leakage in wiretap or obfuscation channels. Figure~\ref{fig:cipher} illustrates a wiretap channel, where Alice wants to communicate effectively with Bob while minimizing the leaked information to Eve. Here, communication channels can be any completely-positive trace-preserving mappings. Alice's strategy in the wiretap channel can be computed by finding an encoding policy $x\mapsto \rho_A^x$ that maximizes $I_{acc}(\{p_X(x),\rho_B^x\}_{x\in\X})$ (i.e., the rate of information transfer to Bob) subject to $\mathcal{Q}(X\rightarrow E)_{\rho_E}\leq \varepsilon$ for small constant $\varepsilon\geq 0$ (i.e., restricts the amount of leaked information to Eve). 
    \item \textit{Generalization in Quantum Machine Learning}: Classical maximal leakage has been already utilized to better understand generalization of machine learning models~\cite{esposito2021generalization,esposito2019new}. Therefore, we expect to be able to use maximal quantum leakage to analyze generalization of various quantum machine learning models.
    \item \textit{Privacy-Preserving Quantum Computing}: Privacy analysis in quantum system is relatively new with only recent studies on quantum differential privacy and puffer-fish privacy~\cite{zhou2017differential,hirche2023quantum,aaronson2019gentle,nuradha2023quantum,farokhi2023privacy}. In many scenarios differential privacy can result in conservative results and bad performance. Information-theoretic privacy~\cite{makhdoumi2014information} can provide a systematic approach to balancing privacy and utility in general settings. The proposed notion of information leakage can provide an operational measure of privacy leakage for balancing utility and privacy in quantum systems.
\end{itemize}

\bibliography{ref}

\appendix

\section{Proofs of All the Presented Results} \label{sec:proofs}
The proofs of all the results in Sections~\ref{sec:leakage}--\ref{sec:noise} are presented across the following subsections. The proofs are moved to this section to minimize interruptions to the flow of the paper and to focus on definitions and properties without excessive mathematics. 
\subsection{Proof of Theorem~\ref{tho:mql}}
\label{proof:tho:mql}
 The proof follows from that
 \begin{align*}
     \sup_{Z,\widehat{Z}} \log_2 &\left(\frac{\mathbb{P}\{Z=\widehat{Z}\}}{\displaystyle \max_{z\in\mathbb{Z}}\mathbb{P}\{Z=z\}} \right)\\
     &=\log_2\left(\sum_{y\in\mathbb{Y}} \max_{x\in\mathbb{X}}\mathbb{P}\{Y=y\,|\,X=x\} \right)\\
     &=\log_2\left(\sum_{y\in\mathbb{Y}} \max_{x\in\mathbb{X}} \trace(\rho_A^x F_y) \right),
 \end{align*}
 where the first equality follows from Theorem~1 in~\cite{issa2019operational} and the second equality is a direct consequence of~\eqref{eqn:conditional}. 

 \subsection{Proof of Theorem~\ref{tho:maxsize}}
 \label{proof:tho:maxsize}
    The proof follows the same line of reasoning as the proof of Theorem~1 in~\cite{davies1978information}, which is reformulated in Proposition~5.8 in~\cite{Holevo2019}.
 
    Let $\mathcal{F}_k$ denote the set of POVMs with $k$ outcomes, i.e., $\mathcal{F}_k:=\{\{F_y\}_{y=1}^k|F_y\succeq 0,\sum_{y=1}^k F_y=I \}$. Define $g:\cup_{k\geq 0} \mathcal{F}_k\rightarrow \mathbb{R}_{\geq 0}$ as
     \begin{align*}
         g(F):=\sum_{y=1}^k \max_{x\in\mathbb{X}} \trace(\rho_A^x F_y).
     \end{align*}
     where $F=\{F_y\}_{y=1}^k$. Evidently, 
     $$
     \mathcal{Q}(X\rightarrow A)_{\rho_A}=\log_2\left(\sup_{F\in \cup_{k\geq 0} \mathcal{F}_k}g(F)\right).
     $$ 
     The following lemma extends Lemma~2 in~\cite{davies1978information} to $g(\cdot)$. 
  
  \begin{lemma} \label{lemma:decompose} Consider $F=\{F_y\}_{y=1}^k\in\mathcal{F}_k$. Let $F':=\{F'_y\}_{y=1}^{k+1}\in\mathcal{F}_{k+1}$ be such that $F'_y=F_y$ for all $y\in\{1,\dots,k\}\setminus\{y_0\}$ for some $y_0\in\{1,\dots,k\}$ and $F'_{y_0},F'_{k+1}\succ 0$ satisfy $F_{y_0}=F'_{y_0}+F'_{k+1}$. Then, $g(F')\geq g(F)$. 
  \end{lemma}

  \begin{proof}
    Without loss of generality, up to rearranging the order of the elements in the POVMs, we can assume that $y_0=k$. Note that
      \begin{align*}
          g(F)
          =&\sum_{y=1}^k \max_{x\in\mathbb{X}} \trace(\rho_A^x F_y) \\
          =&\sum_{y=1}^{k-1} \max_{x\in\mathbb{X}} \trace(\rho_A^x F'_y)+ \max_{x\in\mathbb{X}} \trace(\rho_A^x (F'_{k}+F'_{k+1})) \\
          \leq& \sum_{y=1}^{k-1} \max_{x\in\mathbb{X}} \trace(\rho_A^x F'_y)+ \max_{x\in\mathbb{X}} \trace(\rho_A^x F'_{k})\\
          &+\max_{x\in\mathbb{X}}\trace(\rho_A^x F'_{k+1})\\
          =&\sum_{y=1}^{k+1} \max_{x\in\mathbb{X}} \trace(\rho_A^x F'_y),\\
          =&g(F'),
      \end{align*}
      where the inequality follows from
      \begin{align*}
            \max_{x\in\mathbb{X}} \trace(\rho_A^x &(F'_{k}+F'_{k+1}))\\
            &=\max_{x\in\mathbb{X}} (\trace(\rho_A^x F'_{k})+\trace(\rho_A^x F'_{k+1}))\\
            &\leq \max_{x\in\mathbb{X}} \trace(\rho_A^x F'_{k})+\max_{x\in\mathbb{X}}\trace(\rho_A^x F'_{k+1}),
        \end{align*}
        This concludes the proof.
  \end{proof}

  \begin{lemma} \label{lemma:convex}
      $F\mapsto g(F)$ is a convex on $\mathcal{F}_k$ for any $k\geq 0$.
  \end{lemma}

  \begin{proof} Consider $F=\{F_y\}_{y=1}^k$ and $F'=\{F'_y\}_{y=1}^k$ such that $F,F'\in\mathcal{F}_k$. Evidently, $F'':=\{\lambda F_y+(1-\lambda)F'_y\}_{y=1}^k$ also belongs to $\mathcal{F}_k$ for all $\lambda\in[0,1]$ because $\lambda F_y+(1-\lambda)F'_y\succeq 0$ for all $y\in\{0,\dots,k\}$ and 
  \begin{align*}
      \sum_{y=1}^k (\lambda F_y+(1-\lambda)F'_y)
      =\lambda \sum_{y=1}^k F_y+(1-\lambda)\sum_{y=1}^k F'_y=I.
  \end{align*}
  Furthermore,
  \begin{align*}
      g(F'')
      &=\sum_{y=1}^k \max_{x\in\mathbb{X}} \trace(\rho_A^x (\lambda F_y+(1-\lambda)F'_y))\\
      &=\sum_{y=1}^k \max_{x\in\mathbb{X}} (\lambda\trace(\rho_A^xF_y)+(1-\lambda)\trace(\rho_A^xF'_y))\\
      &\leq \sum_{y=1}^k (\max_{x\in\mathbb{X}} \lambda\trace(\rho_A^xF_y)+\max_{x\in\mathbb{X}}(1-\lambda)\trace(\rho_A^xF'_y))\\
      &=\lambda\sum_{y=1}^k \max_{x\in\mathbb{X}} \trace(\rho_A^xF_y)+(1-\lambda)\sum_{y=1}^k\max_{x\in\mathbb{X}}\trace(\rho_A^xF'_y)\\
      &=\lambda g(F)+(1-\lambda)g(F').
  \end{align*}
    This concludes the proof.
  \end{proof}

      Fix $k\geq 0$. Because $g(F)$ is continuous in $F$ and $\mathcal{F}_k$ is compact, $\sup_{F\in\mathcal{F}_k} g(F)$ is attained on the set $\mathcal{F}_k$. Assume that $\hat{F}=\{\hat{F}_y\}_{y=1}^k$ maximizes $g$ on $\mathcal{F}_k$. By removing zero components, if necessary, we obtain POVM $\tilde{F}=\{\tilde{F}_y\}_{y=1}^\ell$ with $k\geq \ell$ such that $g(\hat{F})=g(\tilde{F})$. Using spectral decomposition, we can decompose each element of the POVM $\tilde{F}$ into a sum of rank-one elements, i.e., $\tilde{F}_y=\sum_{z=1}^{z_y} \bar{F}_z$, where $\bar{F}_z\succ 0$ are rank one matrices. Construct POVM $\bar{F}:=\{\{\bar{F}_z\}_{z=1}^{z_y}\}_{y=1}^\ell$. Let $m=|\bar{F}|$. By repeatedly using Lemma~\ref{lemma:decompose}, we can see that $g(\bar{F})\geq g(\hat{F})$. Because of Lemma~\ref{lemma:convex}, $g$ is convex on $\mathcal{F}_m$. Therefore, the maximizing observable $\check{F}$ is an extreme point of the set $\mathcal{F}_m$. Theorem 2.21 in~\cite{Holevo2019} shows that $\check{F}$ must be linearly independent and Carath\'{e}odory Theorem (see, e.g., Theorem 1.6~in~\cite{Holevo2019}) shows that $\check{F}$ can be represented as a convex combination of at most $d_A^2$ elements, where $d_A$ is the dimension of $\mathcal{H}_A$. Therefore, without loss of generality, by removing some zero components, $\check{F}\in\mathcal{F}_{d_A^2}$. Since $k$ was chosen arbitrary, $\sup_{F\in \cup_{k\geq 0} \mathcal{F}_k}g(F)$ is attained on an observable from the compact convex set $\mathcal{F}_{d_A^2}$. 

\subsection{Proof of Proposition~\ref{prop:independence}: Establishing Independence Property}
    \label{proof:prop:independence}
    Define $I_\infty(X;Y):=\sum_{y\in\mathbb{Y}} \max_{x\in\mathbb{X}} \mathbb{P}\{Y=y\,|\,X=x\}=\sum_{y\in\mathbb{Y}} \max_{x\in\mathbb{X}} \trace(\rho_A^x F_y)$. Note that $I_\infty(X;Y)\geq 0$~\cite[Lemma~1]{issa2019operational} and, as a result, $\mathcal{Q}(X\rightarrow A)_{\rho_A}=\sup_{\{F_y\}_y} I_\infty(X;Y)\geq 0$. 

    Now, we prove that $\mathcal{Q}(X\rightarrow A)_{\rho_A}=0$ if $(\rho_A^x)_{x\in\mathbb{X}}$ is indistinguishable.  Consider $x,x'\in\X$ such that $x\neq x'$. We have 
    \begin{align*}
        \mathbb{P}\{Y=y\,|\,X=x\}
        &=\trace(\rho_A^x F_y)\\
        &=\trace(\rho_A^{x'} F_y)\\
        &=\mathbb{P}\{Y=y\,|\,X=x'\},
    \end{align*}
    where the first and the third equality follow from~\eqref{eqn:conditional} and the second equality follows from indistinguishablity of $(\rho_A^x)_{x\in\mathbb{X}}$, i.e., $\rho_A^x=\rho_A^{x'}$ for all $x,x'\in\X$. Therefore, random variables $X$ and $Y$ are independent, which implies that $I_\infty(X;Y)=0$~\cite[Lemma~1]{issa2019operational}, which is true irrespective of the choice of $\{F_y\}_y$. This implies that $\mathcal{Q}(X\rightarrow A)_{\rho_A}=\sup_{\{F_y\}_y}I_\infty(X;Y)=0$. 

    Next, we prove that $(\rho_A^x)_{x\in\mathbb{X}}$ is indistinguishable if $\mathcal{Q}(X\rightarrow A)_{\rho_A}=0$. Because $\mathcal{Q}(X\rightarrow A)_{\rho_A}=0$, it must be that $I_\infty(X;Y)=0$ for all POVM $F=\{F_y\}_y$. Following~\cite[Lemma~1]{issa2019operational}, it must be that $X$ and $Y$ are independent for all POVM $F$. We have 
    \begin{align*}
        \trace(\rho_A^x F_y)
        &=\mathbb{P}\{Y=y\,|\,X=x\}\\
        &=\mathbb{P}\{Y=y\,|\,X=x'\}\\
        &=\trace(\rho_A^{x'} F_y),
    \end{align*}
    where the second equality follows from statistical independence of $X$ and $Y$. Therefore, $\trace((\rho_A^x-\rho_A^{x'}) F_y)=0$ for all $0\preceq F_y\preceq I$. This implies that $\rho_A^x-\rho_A^{x'}=0$, or equivalently $\rho_A^x=\rho_A^{x'}$, which concludes the proof. 

    \subsection{Proof of Proposition~\ref{property:2}}
    \label{proof:property:2}
    Let $\{F_y\}_y$ be the maximizing POVM in Theorem~\ref{tho:maxsize}. Then, $\mathcal{Q}(X\rightarrow A)=I_\infty(X;Y)$. Following~\cite[Lemma~1]{issa2019operational}, $I_\infty(X;Y)\leq \min\{|\X|,|\Y|\}$. Furthermore, $|\Y|=m\leq d_A^2$. Therefore, $\mathcal{Q}(X\rightarrow A)_{\rho_A}\leq \min\{\log_2(|\X|),\log_2(d_A^2)\}$.

    \subsection{Proof of Proposition~\ref{prop:dataprocessing}: Establishing Data Processing Inequality}
    \label{proof:prop:dataprocessing}
    Note that
    \begin{align*}
        \mathcal{Q}(X\rightarrow A)_{\mathcal{E}(\rho_A)}=\sup_{\{F_y\}\in\mathcal{F}} &\log_2\left(\sum_{y\in\mathbb{Y}} \max_{x\in\mathbb{X}} \trace(\mathcal{E}(\rho_A^x) F_y) \right),
    \end{align*}
    where $\mathcal{F}$ denotes the set of all POVMs. According to Choi-Kraus theorem~\cite[Theorem~4.4.1]{wilde2013quantum}, for each quantum channel $\mathcal{E}$, there exists a family of linear operators $\{E_j\}_{j=1}^n$ for some $n\in\mathbb{N}$ such that $\sum_{j=1}^n E_j^\dag E_j=I$ and $\mathcal{E}(\rho)=\sum_{j=1}^n E_j\rho E_j^\dag$ for all density operators $\rho$. Consider any POVM $F=\{F_y\}_{y=1}^m$. We have
    \begin{align*}
        \trace(\mathcal{E}(\rho_A^x)F_y)
        &=\trace\left(\sum_{j=1}^n E_j\rho_A^x E_j^\dag F_y\right)\\
        &=\trace\left(\rho_A^x \left(\sum_{j=1}^n  E_j^\dag F_yE_j\right)\right)\\
        &=\trace\left(\rho_A^x \bar{F}_y\right),
    \end{align*}
    where 
    \begin{align*}
        \bar{F}_y:=\sum_{j=1}^n  E_j^\dag F_yE_j.
    \end{align*}
    Define $\bar{\mathcal{F}}:=\{\{\bar{F}_y\}_y\,|\,\bar{F}_y=\sum_{j=1}^n   E_j^\dag F_yE_j,\{F_y\}_y\in\mathcal{F}\}$. Evidently,  by construct, 
    \begin{align*}
        \mathcal{Q}(X\rightarrow A)_{\mathcal{E}(\rho_A)}=\sup_{\{\bar{F}_y\}\in\bar{\mathcal{F}}} \log_2\left(\sum_{y\in\mathbb{Y}} \max_{x\in\mathbb{X}} \trace(\rho_A^x \bar{F}_y) \right).
    \end{align*}
    Let us prove that $\bar{\mathcal{F}}\subset \mathcal{F}$. Let $\{\bar{F}_y\}_y\in \bar{\mathcal{F}}$. Then, there must exist $\{F_y\}_y\in\mathcal{F}$ such that $\bar{F}_y=\sum_{j=1}^n   E_j^\dag F_yE_j$. Note that $\bar{F}_y\succeq 0$ because $F_y\succeq 0$. Furthermore, $\bar{F}_y=\sum_{j=1}^n  E_j F_y E_j^\dag \preceq  I$ because $F_y\preceq I$. Furthermore,
    \begin{align*}
        \sum_{y} \bar{F}_y
        &=\sum_{y} \sum_{j=1}^n   E_j^\dag F_yE_j\\
        &= \sum_{j=1}^n   E_j^\dag \left(\sum_{y} F_y \right) E_j
        \\
        &= \sum_{j=1}^n   E_j^\dag E_j\\
        &=I,
    \end{align*}
    where the third equality follows from that $\sum_{y} F_y=I$ and the last equality follows from that $\sum_{j=1}^n E_j^\dag E_j=I$. Therefore, $\{\bar{F}_y\}_{y}$ must belong to $\mathcal{F}$, which proves that $\bar{\mathcal{F}}\subset \mathcal{F}$. This implies that
    \begin{align*}
        \mathcal{Q}(X\rightarrow A)_{\mathcal{E}(\rho_A)}
        &=\sup_{\{\bar{F}_y\}\in\bar{\mathcal{F}}} \log_2\left(\sum_{y\in\mathbb{Y}} \max_{x\in\mathbb{X}} \trace(\rho_A^x \bar{F}_y) \right)\\
        &\leq \sup_{\{\bar{F}_y\}\in \mathcal{F}} \log_2\left(\sum_{y\in\mathbb{Y}} \max_{x\in\mathbb{X}} \trace(\rho_A^x \bar{F}_y) \right)\\
        &=\mathcal{Q}(X\rightarrow A)_{\rho_A},
    \end{align*}
    where the inequality follows from that taking the supremum over a larger set results in a larger value. 
    

    \subsection{Proof of Proposition~\ref{prop:relation_accessible_info}: Relationship Between Accessible Information and Maximal Quantum Leakage}
    \label{proof:prop:relation_accessible_info}
    Note that $I_\infty(X;Y):=\sum_{y\in\mathbb{Y}} \max_{x\in\mathbb{X}} \mathbb{P}\{Y=y\,|\,X=x\}=\sum_{y\in\mathbb{Y}} \max_{x\in\mathbb{X}} \trace(\rho_A^x F_y)$~\cite[Theorem~1]{issa2019operational}. From~\cite[Lemma~2]{issa2019operational}, we know that $I_\infty(X;Y)\geq I(X;Y)$. The rest follows from taking supremum with respect to POVM $\{F_y\}_y$ on both sides of the inequality $I_\infty(X;Y)\geq I(X;Y)$.

\subsection{Proof of Proposition~\ref{prop:global}: Effect of Global Depolarizing Channel}
\label{proof:prop:global}
Note that
\begin{align*}
        \trace(\mathcal{D}_{p,d_A}(\rho_A^x)F_y)
        &=\trace\left(\left(\frac{p}{d_A}I+(1-p)\rho_A^x\right)F_y\right)\\
        &=\frac{p}{d_A}\trace(F_y) +(1-p)\trace\left(\rho_A^x F_y\right).
    \end{align*} 
Therefore, 
    \begin{align*}
     \sum_{y\in\mathbb{Y}} &\max_{x\in\mathbb{X}} \trace(\mathcal{D}_{p,d_A}(\rho_A^x)F_y)\\
     &=\sum_{y\in\mathbb{Y}} \max_{x\in\mathbb{X}} \left( \frac{p}{d_A}\trace(F_y) +(1-p)\trace\left(\rho_A^x F_y\right)\right)\\
     &=  \frac{p}{d_A}\sum_{y\in\mathbb{Y}}\trace(F_y) +(1-p)\sum_{y\in\mathbb{Y}}\max_{x\in\mathbb{X}}\trace\left(\rho_A^x F_y\right)\\
     &=  \frac{p}{d_A}\trace\bigg(\sum_{y\in\mathbb{Y}} F_y\bigg) +(1-p)\sum_{y\in\mathbb{Y}}\max_{x\in\mathbb{X}}\trace\left(\rho_A^x F_y\right)\\
     &=  \frac{p}{d_A}\trace(I) +(1-p)\sum_{y\in\mathbb{Y}}\max_{x\in\mathbb{X}}\trace\left(\rho_A^x F_y\right)\\
     &=  p +(1-p)\sum_{y\in\mathbb{Y}}\max_{x\in\mathbb{X}}\trace\left(\rho_A^x F_y\right).
    \end{align*}
This implies that
    \begin{align*}
        \mathcal{Q}(X\rightarrow A&)_{\mathcal{D}_{p,d_A}(\rho_A)}\\
        &=\log_2\bigg(\sup_{\{F_y\}_y}\sum_{y\in\mathbb{Y}} \max_{x\in\mathbb{X}} \trace(\mathcal{D}_{p,d_A}(\rho_A^x)F_y)\bigg)\\
        &=\log_2\bigg( p +(1-p)\sup_{\{F_y\}_y} \sum_{y\in\mathbb{Y}}\max_{x\in\mathbb{X}}\trace\left(\rho_A^x F_y\right)\!\!\bigg)\\
        &=\log_2( p +(1-p)2^{\mathcal{Q}(X\rightarrow A)_{\rho_A}}).
    \end{align*}

\subsection{Proof of Proposition~\ref{prop:local}: Effect of Local Depolarizing Channel}
\label{proof:prop:local}
Following the proof of Lemma~IV.4 in~\cite{hirche2023quantum}, a local depolarizing noise channel can be always represented as
\begin{align*}
    \mathcal{D}_{p,2}^{\otimes k}=p^k\frac{I}{2^k}+(1-p^k)\mathcal{M}(p),
\end{align*}
where $\mathcal{M}$ is an appropriately selected quantum channel, i.e., completely positive and trace preserving mapping. Using the same line of reasoning as in the proof of Proposition~\ref{prop:global} in Subsection~\ref{proof:prop:global}, we get 
\begin{align}
        \nonumber\mathcal{Q}&(X\rightarrow A)_{\mathcal{D}_{p,2}^{\otimes k}(\rho_A)}\\
        &=\log_2\bigg( p^k\! +\!(1\!-\!p^k)\!\sup_{\{F_y\}_y} \!\sum_{y\in\mathbb{Y}}\max_{x\in\mathbb{X}}\trace\left(\mathcal{M}(\rho_A^x) F_y\right)\!\!\bigg).\label{eqn:proof:1}
    \end{align}
From Proposition~\ref{prop:dataprocessing}, we know that 
\begin{align}
\nonumber\sup_{\{F_y\}_y} \sum_{y\in\mathbb{Y}}&\max_{x\in\mathbb{X}}\trace\left(\mathcal{M}(\rho_A^x) F_y\right)
\\&\leq 
\sup_{\{F_y\}_y} \sum_{y\in\mathbb{Y}}\max_{x\in\mathbb{X}}\trace\left(\rho_A^x F_y\right).\label{eqn:proof:2}
\end{align}
Combining~\eqref{eqn:proof:1} and~\eqref{eqn:proof:2} gives
\begin{align*}
        \mathcal{Q}&(X\rightarrow A)_{\mathcal{D}_{p,2}^{\otimes k}(\rho_A)}\\
        &\leq\log_2\bigg( p^k +(1-p^k)\sup_{\{F_y\}_y} \sum_{y\in\mathbb{Y}}\max_{x\in\mathbb{X}}\trace\left(\rho_A^x F_y\right)\!\!\bigg)\\
        &=\log_2( p^k +(1-p^k)2^{\mathcal{Q}(X\rightarrow A)_{\rho_A}}).
\end{align*}

\end{document}